\documentclass[num-refs,serif]{wiley-article}

\usepackage{amsmath}
\usepackage{cleveref}
\usepackage{subfigure}
\usepackage{todonotes}
\usepackage{comment}
\usepackage{tikz}
\usepackage{enumitem}
\usepackage{lineno}
\usepackage{algorithm}
\usepackage[noend]{algpseudocode}
\usepackage{anyfontsize}
\usepackage{comment}

\setlist[enumerate]{noitemsep, nolistsep}
\setlist[itemize]{noitemsep, nolistsep}

\usetikzlibrary{shapes.geometric}

\newcommand{\N}{\mathcal{N}}
\newcommand{\NP}{\mathcal{NP}}

\papertype{Original Article}
\paperfield{Journal Section}

\title{On $\lambda$-backbone coloring of cliques with tree backbones in linear time}

\author[1]{Krzysztof Michalik}
\author[1]{Krzysztof Turowski}

\affil[1]{Theoretical Computer Science Department, Jagiellonian University, Kraków, 30-348 Poland}

\corraddress{Krzysztof Turowski, Theoretical Computer Science Department, Jagiellonian University, Krakow, 30-348 Poland}
\corremail{krzysztof.szymon.turowski@gmail.com}

\fundinginfo{Krzysztof Turowski work was supported by Polish National Science Center $2020$/$39$/D/ST$6$/$00419$ grant. For the purpose of Open Access, the author has applied a CC-BY public copyright license to any Author Accepted Manuscript (AAM) version arising from this submission.}

\runningauthor{Krzysztof Michalik et al.}

\begin{document}

\begin{frontmatter}
\maketitle

\begin{abstract}
  A $\lambda$-backbone coloring of a graph $G$ with its subgraph (also called a \emph{backbone}) $H$ is a function $c \colon V(G) \rightarrow \{1,\dots, k\}$ ensuring that $c$ is a proper coloring of $G$ and for each $\{u,v\} \in E(H)$ it holds that $|c(u) -  c(v)| \ge \lambda$.

  In this paper we propose a way to color cliques with tree and forest backbones in linear time that the largest color does not exceed $\max\{n, 2 \lambda\} + \Delta(H)^2 \lceil\log{n} \rceil$.
  This result improves on the previously existing approximation algorithms as it is $(\Delta(H)^2 \lceil\log{n} \rceil)$-absolutely approximate, i.e. with an additive error over the optimum.
  
  We also present an infinite family of trees $T$ with $\Delta(T) = 3$ for which the coloring of cliques with backbones $T$ requires at least $\max\{n, 2 \lambda\} + \Omega(\log{n})$ colors for $\lambda$ close to $\frac{n}{2}$. The construction draws on the theory of Fibonacci numbers, particularly on Zeckendorf representations. 
\end{abstract}
\end{frontmatter}

\section{Introduction}

Within the general framework of graph coloring problems, there exists an important class of problems that is related to the frequency assignment problem: for a given set of transmitters (which are represented by the vertices
of a graph) and their adjacency (i.e. adjacent transmitters are close
enough or have a signal which is strong enough), assign the frequency bands to the transmitters in a way that keeps interference below a defined level while minimizing the total frequency span.
In some applications, it makes sense to distinguish a certain substructure of the network (called the \emph{backbone}), designed as the most important part of communication, so it has to meet additional restrictions on the assignment.

This leads us to the backbone coloring problem, introduced by Broersma in \cite{broersma2007backbone}. First, let us define formally the backbone coloring, using a standard graph notation, e.g. from \cite{murty2008graph}:
\begin{definition}
    Let $G$ be a graph and $H$ be a subgraph of $G$ with $V(G) = V(H)$. Let also $\lambda \in \mathbb{N}_+$, $\lambda \ge 2$.
    The \emph{$\lambda$-backbone coloring} of $G$ with backbone $H$ is defined as a function $c\colon V(G) \to \mathbb{N}_+$ such that
    \begin{itemize}
        \item $c(u) \neq c(v)$ for every $uv \in E(G)$,
        \item $|c(u) - c(v)| \ge \lambda$ for every $uv \in E(H)$.
    \end{itemize}
\end{definition}
Note that it differs from the vertex coloring problem in important ways: in an optimal $\lambda$-backbone coloring (i.e. using a minimum number of colors) the ordering of the colors matters, therefore we might observe that some smaller colors are not used while the larger ones are in use.

Broersma et. al. in \cite{broersma2007backbone} also defined the backbone coloring problem as an extension of a classical vertex coloring problem:
\begin{definition}
    Let $G$ be a graph and $H$ be a subgraph of $G$ with $V(G) = V(H)$. Let also $\lambda$ be a natural number greater than $1$.
    The \emph{$\lambda$-backbone coloring problem} is defined as following:
    for a given positive integer $k$, does there exists $c\colon V(G) \to [1, k]$ such that $c$ is a $\lambda$-backbone coloring of $G$ with backbone $H$?
\end{definition}
The \emph{$\lambda$-backbone coloring number} for a graph $G$ with backbone $H$ (denoted as $BBC_\lambda(G, H)$) is then defined as the smallest $k$ such that there exists a $\lambda$-backbone coloring of $G$ with backbone $H$.

Note that here and throughout the whole paper we rely on the notation $[a, b]$ representing a set of integers between $a$ and $b$, inclusive.

In general, it is straightforward to prove that the problem of determining $BBC_\lambda(G, H)$ is $\NP$-hard \cite{broersma2007backbone}.
In the same paper which introduced this problem there were shown basic bounds on the value of the $\lambda$-backbone coloring number, depending on $\chi(G)$, the chromatic number of $G$:
\begin{theorem}[Broersma et al., \cite{broersma2007backbone}]
\label{thm:broersma}
Let $G$ be a graph and $H$ its spanning subgraph.
Then $\chi(G) \le BBC_{\lambda}(G, H) \le \lambda(\chi(G) - 1) + 1$.
\end{theorem}

In \cite{havet2014circular} and \cite{janczewski2015backbone} there were proposed other bounds, suited particularly for graphs with $\chi(H) \ll \chi(G)$:
\begin{theorem}[Havet et al., \cite{havet2014circular}]
\label{thm:havet}
Let $G$ be a graph and $H$ its spanning subgraph.
Then $BBC_{\lambda}(G, H) \le (\lambda + \chi(G) - 2) \chi(H) - \lambda + 2$.
\end{theorem}

\begin{theorem}[Janczewski, Turowski, \cite{janczewski2015backbone}]
Let $G$ be a graph on $n$ vertices and $H$ its spanning subgraph. Then $\lambda(\chi(H) - 1) + 1 \le BBC_{\lambda}(G, H) \le \lambda(\chi(H) - 1) + n - \chi(H) + 1$.
\end{theorem}

The $\lambda$-backbone coloring problem was studied for several classes of graphs, for example split graphs \cite{broersma2009backbone}, planar graphs \cite{havet2014circular}, complete graphs \cite{janczewski2015computational}, and for several classes of backbones: matchings and disjoint stars \cite{broersma2009backbone}, bipartite graphs \cite{janczewski2015computational} and forests \cite{havet2014circular}.
For a special case $\lambda = 2$ it was also studied for many other cases e.g. triangle-free graphs with tree backbones \cite{mivskuf2009backbone}, outerplanar graphs with matching backbones \cite{araujo2019backbone} and general graphs with bounded degree backbones \cite{mivskuf2010backbone} or with tree and path backbones \cite{broersma2007backbone}.

In this paper, we turn our attention to the special case when the graph is complete (denoted $K_n$) and its backbone is a (nonempty) tree or a forest (which we will denote by $T$ and $F$, respectively).
Note that it has a natural interpretation as a labeling problem: how to assign different labels to all vertices such that on every backbone edge the difference between labels is at least $\lambda$.
This description draws a comparison e.g. to $L(k, 1)$-labeling problem (see e.g. \cite{calamoneri2011h} for a survey), where the colors of any two adjacent vertices have to differ by at least $k$ and the colors of any two vertices within distance $2$ have to be distinct.

Broersma in \cite{broersma-general} pointed out that the backbone coloring problem generalizes so called radio labeling problem: finding for a given graph $G$ the function $L\colon V(G) \to \N_+$ such that for every edge $uv \in E(G)$ it holds that $|L(u) - L(v)| \ge 2$ such that $RL(G) = \max_{v \in V(G)} L(v)$ is minimal -- that is, it is exactly the case that $RL(G) = BBC_2(K_{|V(G)|}, G)$.
However, note that the name \emph{radio labeling} sometimes it refers to a different, although related problem \cite{cada,bantva,korvze}.

In \cite{fotakis-spirakis} it was proved that the radio labeling problem can be solved in polynomial time for graphs for which $k$-coloring can be found in polynomial time for some fixed $k$, e.g. for planar graphs or graphs with bounded treewidth.
Additionally, \cite{damaschke} proved for comparability graphs we can find a partition of $V(G)$ into at most $k$ sets which induce semihamiltonian subgraphs in the complement of $G$ (i.e. it contains a Hamiltonian path) and from that it follows that $BBC_2(K_n, G)$ problem is solvable in polynomial time if $G$ is a comparability graph \cite{broersma-general}.
This generalizes an earlier result from \cite{chang-kuo}, where (under the name of $L'(2, 1)$-labelings) it was shown that this problem is polynomially solvable on cographs, a subclass of comparability graphs.

On the hardness side, one has to point out that the radio labeling problem is equivalent to the travelling salesman problem with distances $1$ (on the edges of $G$) and $2$ (on all the other edges). Thus, in general it is MAX SNP-hard \cite{papadimitriou} and $(\frac{535}{534} - \varepsilon)$-inapproximable (unless P = NP) for any $\varepsilon > 0$ \cite{karpinski-approx}, and it remains hard even for graphs of diameter $2$ \cite{fotakis-spirakis}, though it is $\frac{8}{7}$-approximable in polynomial time \cite{karpinski-approx,adamaszek}.

Returning to the backbone coloring problem with arbitrary $\lambda$, it is obvious that $\chi(K_n) = n$, $\chi(F) = 2$, so \Cref{thm:broersma,thm:havet} combined together give us the following bounds:
\begin{align*}
    \max\{n, \lambda + 1\} \le BBC_{\lambda}(K_n, F) \le \lambda + n - 1.
\end{align*}
Moreover, it was proved before in \cite{janczewski2015backbone} that there exists a $2$-approximate algorithm for complete graphs with bipartite backbones and a $3/2$-approximate algorithm for complete graphs with connected bipartite backbones. Both algorithms run in linear time. As a corollary, it was proved that we can compute $BBC_\lambda(K_n, F)$ in quadratic time, provided that $F$ is a tree on $n$ vertices and $\lambda > n - 2$.
On the other hand, for $\lambda = 2$ we know that $BBC_\lambda(K_n, F) \le n + 1$ and moreover we can solve the problem in polynomial time \cite[Theorem 3]{turowski2015optimal}.

Still, this leaves the complexity status of computing $BBC_\lambda(K_n, F)$ or its approximability in polynomial time for arbitrary $\lambda$ open. This paper aims to narrow this gap by providing an improved algorithm for $BBC_\lambda(K_n, F)$ with an additive error.

We start \Cref{sec:positive} by proving that if $F$ is a tree or forest on $n$ vertices with a maximum degree $\Delta(F)$, then $BBC_{\lambda}(K_n, F) \le \max\{n, 2 \lambda\} + \Delta^2(F) \lceil\log{n}\rceil$.
Note that this bound can be much better than the previously known ones, especially for $\lambda$ close to $\frac{n}{2}$ and small $\Delta(F)$.
We also provide a polynomial (i.e. linear for trees, quadratic for forests) algorithm to find the respective $\lambda$-backbone coloring.

This, in turn, combined with another much simpler algorithm allows us to show that we can find in polynomial time a $\lambda$-backbone coloring for $G$ with backbone forest $F$ that uses at most $\Delta^2(F) \lceil\log{n}\rceil$ colors more than the optimal $\lambda$-backbone coloring.

Previously it was known that $BBC_{\lambda}(K_n, T) = \lambda + n - 1$ when $T$ is a star. However, one can ask a more general question: how large can $BBC_{\lambda}(K_n, T)$ be when $\Delta(T)$ is small?
In \Cref{sec:negative} we show that there exists a family of trees with $\Delta(T) = 3$ such that $BBC_\lambda(K_n, T) \ge \max\{n, 2 \lambda\} + \frac{1}{48} \log_\phi{n} - 3$ (with $\phi = \frac{1 + \sqrt{5}}{2}$ being the golden ratio constant).
This result is complementary to the one in the previous section, as it shows that sometimes we need up to $\max\{n, 2 \lambda\} + \Theta(\log{n})$ colors even when we have $\Delta(T) = 3$. In a sense, we might also say that the logarithmic loss over a trivial lower bound $BBC_\lambda(K_n, T) \ge \max\{n, \lambda + 1\}$ for backbone coloring of $K_n$ with a tree backbone $T$ is tight.

Since it was proved in \cite{araujo2017existence} that for every $\lambda \ge 2$ and every connected graph $G$ there exists a spanning tree $T$ such that $BBC_\lambda(G, T) = \chi(G)$, our result establishes additionally that depending on the choice of a backbone tree, the value of $BBC_\lambda(K_n, T)$ can vary considerably, from $n$ to $\max\{n, 2 \lambda\} + \Theta(\log{n})$.
 
Finally, \Cref{sec:open} concludes with a presentation of some open problems related to our work.

\section{Complete graph with a tree or forest backbone: an algorithm}
\label{sec:positive}

In this section we will proceed as follows: we first introduce the so-called red-blue-yellow $(k,l)$-decomposition of a forest $F$ on $n$ vertices, which finds a set $Y$ of size at most $l$ such that we can split $V(F) \setminus Y$ into two independent and almost equal sets $R$, $B$ (with $|R| - |B| \le k$).
Then, we show that we can color $R$ and $B$ with sets of consecutive colors and assign to vertices from $Y$ the smallest and the largest colors in a way that in total we find a $\lambda$-backbone coloring in which the maximum color does not exceed $\max\{n, 2 \lambda\}$ more than $\Delta^2(F) \log{n}$.

We start with a few remarks on notation and some definitions.
Let $c$ be a unique $2$-coloring (up to the permutations of colors) of a tree $T$.
Let us also define $C_i(T) = \{v \in V(T)\colon c(v) = i\}$, i.e. the number of vertices in color $i$, for $i = 1, 2$.
Assume without loss of generality that $c$ is such that $
|C_1(T)| \ge |C_2(T)|$.

Throughout the paper we will use the concept of tree imbalance, defined formally as:
\begin{definition}
    Let $T$ be a tree and $c$ its unique $2$-coloring with $|C_1(T)| \ge |C_2(T)|$.
    Let the \emph{imbalance number} of $T$ be $imb(T) := |C_1(T)| - |C_2(T)|$.
\end{definition}
From this definition, it directly follows that $imb(T) \ge 0$.

First, let us prove the simple structural fact about trees in general:
\begin{lemma}
    \label{lem:tree_half}
    In every tree $T$ there exists a \emph{central vertex} $v \in V(T)$ such that every connected component of $T - v$ has at most $\frac{|V(T)|}{2}$ vertices.
\end{lemma}

\begin{proof}
    The idea is to start from any vertex $w$, and then jump to its neighbor with the largest component size in $T - w$, until we hit a vertex with desired property.
    Note that for any vertex $v$ there can be at most one neighbor $u$ such that its connected component $T_u$ in $T - v$ has more than $\frac{|V(T)|}{2}$ vertices, so the jumps are unique.

    Suppose that the number of vertices in $T_u$ is equal to $k$. Then the number of vertices in all components of $T - v$ sum to $|V(T)| - k - 1$. Now if we look at $T - u$, then it has
    \begin{itemize}
        \item one connected component containing $v$ with $(|V(T)| - k - 1) + 1 < \frac{|V(T)|}{2}$ vertices,
        \item all other components which are subtrees of $T_u - u$, so they contain at most $k - 1$ vertices. 
    \end{itemize}
    Thus, every jump reduces the size of the largest connected component by at least $1$ -- and therefore the algorithm always terminates correctly.

\end{proof}

\begin{definition}
    \label{def:rby}
   Let $T$ be a tree on $n$ vertices.
   We call a partition $(R, B, Y)$ of $V(T)$ a \emph{red-blue-yellow} $(k, l)$-decomposition of $T$
   if $|Y| \le l$ and $R$ and $B$ are independent sets with $|R| - |B| = k$.
\end{definition}
From this definition it follows that if $(R, B, Y)$ is a \emph{red-blue-yellow} $(k, l)$-de\-com\-po\-si\-tion, then $(B, R, Y)$ is a \emph{red-blue-yellow} $(-k, l)$-decomposition.

\begin{lemma}
    \label{lem:imbalance}
    If $T$ is a tree with a root $v \in V(T)$ and $T_1, \dots, T_k$ are subtrees of $T$, made by removing $v$ from $T$, then $imb(T) \le 1 + \sum_{i=1}^k{imb(T_i)}$.
\end{lemma}

\begin{proof}
    Let us denote by $v_i$ the root of $T_i$ for $i = 1, \ldots, k$.
    Observe that
    \begin{itemize}
        \item either $C_1(T_i) \subseteq C_1(T)$ and $C_2(T_i) \subseteq C_2(T)$ (when $v_i \in C_2(T)$),
        \item or $C_1(T_i) \subseteq C_2(T)$ and $C_2(T_i) \subseteq C_1(T)$ (when $v_i \in C_1(T)$).
    \end{itemize}
    Thus, the imbalance of $T$ can be expressed in the following formula (where $1$ appears because of the vertex $v$):
    \begin{align*}
        imb(T) = |C_1(T)| - |C_2(T)| \le 1 + \sum_{i=1}^k \left(|C_1(T_i)| - |C_2(T_i)|\right) = 1 + \sum_{i=1}^k{imb(T_i)}.
    \end{align*}
\end{proof}

\begin{lemma}
    \label{lem:rby-tree}
    For any tree $T$ on $n$ vertices and for any $k \in [0, imb(T)]$ there exists a red-blue-yellow $(k, \lceil \log{n} \rceil)$-de\-com\-po\-si\-tion of $T$.
\end{lemma}

\begin{proof}
    We prove existence of such a decomposition by constructing it.

    We begin by assigning $R_0, B_0, Y_0, T_0$ as initial sets of $R, B, Y, T$ and for convenience, we initialize $R_0$ with a set $D$ of $k$ dummy isolated vertices, while $B_0 = Y_0 = \emptyset$ and $T_0 = V(T)$.
    Note, that any red-blue-yellow $(0, \lceil \log{n} \rceil)$-decomposition of $T \cup D$ can be directly turned into a red-blue-yellow $(k, \lceil \log{n} \rceil)$-decomposition of $T$ just by removing vertices from $D$ and possibly swapping first two sets.

    That is why, from now on, we will concentrate on balancing both of the sets, so that they end up with the same amount of vertices. For that purpose, we will be constructing iteratively sets $R_i, B_i, Y_i, T_i$ and introduce invariants that will hold at each step while decomposing $V(T)$ into $R,B,Y$, which guarantee the convergence to a desired solution:
    \begin{enumerate}
        \item $V(T) \cup D = R_i \cup B_i \cup Y_i \cup V(T_i)$,\label{invariant:1}
        \item $T_i$ is a tree,\label{invariant:2}
        \item $R_i$ and $B_i$ are independent sets,\label{invariant:3}
        \item $T_i$ does not have neighbors in $R_i \cup B_i$,\label{invariant:4}
        \item $0 \le |R_i| - |B_i| \le imb(T_i)$,\label{invariant:5}
    \end{enumerate}
    Of course, all of the above conditions hold for $i = 0$.

    To construct $(R_{i+1}, B_{i+1}, Y_{i+1}, T_{i+1})$ from $(R_i, B_i, Y_i, T_i)$, we do as follows:

    Consider vertex $v$ that we find using \Cref{lem:tree_half} in $T_i$ and denote by $T_i^j$ ($j = 1, \ldots, d$) the subtrees obtained by removing $v$ from $T_i$. Let us consider two cases, depending on which of the values $|R_i| - |B_i|$ and $\sum_{j = 1}^d imb(T_i^j)$ is larger.
    
    \emph{Case (A):} if it holds that $|R_i| - |B_i| > \sum_{j = 1}^d imb(T_i^j)$. However from invariant (\ref{invariant:5}) we also know that $imb(T_i) \ge |R_i| - |B_i|$. Iherefore $imb(T_i) > \sum_{j = 1}^d imb(T_i^j)$, but from \Cref{lem:imbalance} we know that $imb(T_i) \le 1 + \sum_{j = 1}^d imb(T_i^j)$, so it has to be the case that $1 + \sum_{j = 1}^d imb(T_i^j) = imb(T_i) = |R_i| - |B_i|$
    Then, we observe that for $R = R_i \cup C_2(T_i)$, $B = B_i \cup C_1(T_i)$, and $Y  = Y_i$ it is true that:
    \begin{itemize}
        \item $(R, B, Y)$ is a partition of $V(T) \cup D$,
        \item $R$ and $B$ are both independent sets -- due to the invariants (\ref{invariant:3}) and (\ref{invariant:4}),
        \item $|R| - |B| = (|R_i| - |B_i|) - (|C_1(T_i)| - |C_2(T_i)|) = imb(T_i) - imb(T_i) = 0$,
        \item $|Y_i| = i \le \lceil\log_2{n}\rceil$.
    \end{itemize}
    Therefore, $(R, B, Y)$ is a red-blue-yellow $(0, \lceil \log{n} \rceil)$-decomposition of $T \cup D$.
    
    \emph{Case (B):} if it is true that $|R_i| - |B_i| \le \sum_{j = 1}^d imb(T_i^j)$.

    Until now we never needed to order subtrees in any way, so we can safely assume that they are sorted nondecreasingly according to their imbalance (i.e. $imb(T_i^1)\le \dots \le imb(T_i^d)$). We will be iterating through $T_i^j$ and put them in either of the two sets, $R$ or $B$.

    Let us define intermediate sets $R_i^j, B_i^j$ as sets that we get after we processed first $j$ subtrees. Clearly $R_i^0 = R_i$ and $B_i^0 = B_i$.

    To get next pair of sets $R_i^{j+1}, B_i^{j+1}$ we will follow a very simple rule - we choose a smaller set (say it is $B_i^j$) and add $C_1(T_i^{j+1})$ to it, while we add $C_2(T_i^{j+1})$ to $R_i^j$. If sets are equal, we choose one of them arbitrarily. If necessary, we swap sets, so $|R_i^{j+1}| > |B_i^{j+1}|$ holds after each iteration. We iterate this procedure for $j=1,\dots, d-1$, while the final tree, $T_i^d$ becomes our $T_{i+1}$, the final sets $R_i^{d-1}$ and $B_i^{d-1}$ become our $R_{i + 1}$ and $B_{i + 1}$, and we proceed to the next iteration.

    Let us now prove that the invariants (\ref{invariant:1})--(\ref{invariant:5}) are preserved between iterations.
    The invariant (\ref{invariant:1}) holds directly by the construction of the respective sets:
    \begin{itemize}
        \item we assign $v_i$ to $Y_{i + 1}$,
        \item we distribute all vertices of subtrees $T_i^j$ for $j = 1, \ldots, d - 1$ between $R_{i + 1}$ and $B_{i + 1}$,
        \item the remaining subtree $T_i^d$ becomes the new $T_{i + 1}$.
    \end{itemize}
    Moreover, the invariant (\ref{invariant:2}) is trivially preserved directly by the fact that $T_{i + 1}$ is a subtree of $T_i$.

    The invariant (\ref{invariant:3}) states that $R_i$ and $B_i$ are independent sets. From the invariant (\ref{invariant:4}) it follows that if add to them any independent sets from $T_i$, they still remain independent. Thus, since $C_1(T_i^j)$ and $C_2(T_i^j)$ are independent, and there is no edge between $C_l(T_i^j)$ and $C_l(T_i^{j'})$ for any $l = 1, 2$ and $j \neq j'$, it follows that $R_{i + 1}$ and $B_{i + 1}$ are also independent sets -- and therefore the invariant (\ref{invariant:3}) is true also for $i + 1$.

    By the invariant (\ref{invariant:4}) for $i$, we know that $T_i$ has its neighbors only in $Y_i$. Thus, $T_{i + 1}$, is not only a subtree $T_i$, but it is only adjacent to $Y_i$ or $v$ itself, that is, only to $Y_{i + 1}$ -- thus the invariant (\ref{invariant:4}) holds for $i + 1$.
    
    The crucial observation to prove the invariant (\ref{invariant:5}) is that after $j$-th iteration of the inner for loop we have
    \begin{align*}
        0 \le |R_i^j| - |B_i^j| \le \sum_{j' = j + 1}^d imb(T_i^{j'}).
    \end{align*}
    Clearly this holds for $j = 0$ as this condition is identical to $0 \le |R_i| - |B_i| \le \sum_{j = 1}^d imb(T_i^j)$. To prove it inductively for $j = 1, 2, \ldots, d - 1$, we distinguish two cases:
    \begin{itemize}
        \item if $|R_i^{j - 1}| - |B_i^{j - 1}| < imb(T_i^d)$, then
            \begin{align*}
                |R_i^j| - |B_i^j|
                    & = \left||R_i^{j - 1}| + |C_2(T_i^j)| - |B_i^{j - 1}| - |C_1(T_i^j)|\right| 
                      = \left||R_i^{j - 1}| - |B_i^{j - 1}| - imb(T_i^j)\right| \\
                    & \le \max\{|R_i^{j - 1}| - |B_i^{j - 1}|, imb(T_i^j)\} \le imb(T_i^d),
            \end{align*}
        \item if $|R_i^{j - 1}| - |B_i^{j - 1}| \ge imb(T_i^d)$, then
            \begin{align*}
                |R_i^j| - |B_i^j|
                    & = |R_i^{j - 1}| + |C_2(T_i^j)| - |B_i^{j - 1}| - |C_1(T_i^j)| 
                      = |R_i^{j - 1}| - |B_i^{j - 1}| - imb(T_i^j) \\
                    & \le \sum_{j' = j}^d imb(T_i^{j'}) - imb(T_i^j) = \sum_{j' = j + 1}^d imb(T_i^{j'}),
            \end{align*}
            where the second to last inequality follows from the induction assumption that $|R_i^{j - 1}| - |B_i^{j - 1}| \le \sum_{j' = j}^d imb(T_i^{j'})$.
    \end{itemize}
    Ultimately, for $j = d - 1$ we get in either case that $|R_{i + 1}| - |B_{i + 1}| = |R_i^{d - 1}| - |B_i^{d - 1}| \le imb(T_i^d) = imb(T_{i + 1})$, and clearly by construction it is always true that $|R_i^j| - |B_i^j| \ge 0$ so the invariant (\ref{invariant:5}) is also preserved for $i + 1$.

    The above inductive proof establishes that all the invariants are true at the beginning of the iteration when $V(T_k) = \emptyset$.
    Therefore, from invariant (\ref{invariant:1}) it follows that $(R_k, B_k, Y_k)$ is a partition of $V(T) \cup D$.
    From invariants (\ref{invariant:1}) and (\ref{invariant:3}) it follows that $0 \le |R_k| - |B_k| \le imb(T_k) = 0$ and that $R_i$ and $B_i$ are independent sets, respectively.
    Finally, from \Cref{lem:tree_half} we know that $\frac{|V(T_i)|}{2} \le \frac{|V(T_{i + 1})|}{2}$ for all $i = 0, 1, \ldots$, so $k \le \lceil\log{n}\rceil$ and therefore $|Y_k| \le \lceil\log{n}\rceil$
    Thus, $(R_k, B_k, Y_k)$ is a red-blue-yellow $(0, \lceil \log{n} \rceil)$-decomposition of $T \cup D$.
\end{proof}

\begin{algorithm}[ht]
\begin{algorithmic}[1]
  \Function{\texttt{RBY-decompose}}{$T$, $k$}
    \State $D \gets$ a set of dummy $k$ isolated vertices
    \State $T_0 \gets T$, $R_0 \gets D$, $B_0 \gets \emptyset$, $Y_0 \gets \emptyset$
    \For{$i = 0, 1, 2, \ldots$}
      \If{$V(T_i) == 0$}
        \State $R \gets R_i$, $B \gets B_i$, $Y \gets Y_i$
        \State \textbf{break}
      \EndIf
      \State Find a central vertex $v_i$ in $T_i$ using \Cref{lem:tree_half}
      \State $T'_1, \ldots, T'_d \gets$ subtrees of $T_i$ rooted in $v$ with nondecreasing $imb(T'_j)$
      \If{$|R_i| - |B_i| > \sum_{j = 1}^d imb(T'_j)$}
        \State $R \gets R_i \cup C_2(T_i)$, $B \gets B_i \cup C_1(T_i)$, $Y \gets Y_i$
        \State \textbf{break}
      \Else
        \State $R_i(0) \gets R_i$, $B_i(0) \gets B_i$
        \For{$j = 1, 2, \ldots, d - 1$}
          \State $R_i(j) \gets R_i(j - 1) \cup C_2(T'_j)$, $B_i(j) \gets B_i(j - 1) \cup C_1(T'_j)$
          \If{$|R_i(j)| < |B_i(j)|$}
            \State Swap $R_i(j)$ and $B_i(j)$
          \EndIf
        \EndFor
        \State $R_{i + 1} \gets R_i(d - 1)$, $B_{i + 1} \gets B_i(d - 1)$, $Y_{i + 1} \gets Y_i \cup \{v_i\}$, $T_{i + 1} \gets T'_d$
      \EndIf
    \EndFor
    \If{$|R \setminus D| < |B \setminus D|$}
      \State Swap $R$ and $B$
    \EndIf
    \State \Return $(R \setminus D, B \setminus D, Y)$
  \EndFunction 
\end{algorithmic}
\caption{Find a red-blue-yellow $(k, \lceil \log{n} \rceil)$-decomposition of $T$.}
\label{alg:rby}
\end{algorithm}

Note that the proof above can be directly translated to \Cref{alg:rby}.
In addition,
\begin{theorem}
    \Cref{alg:rby} with an appropriate preprocessing runs in $O(n)$ time.
\end{theorem}

\begin{proof}
    Observe that for a tree on $n$ vertices we can compute for every vertex $v$ and its neighbor $u$ functions $f(v, u)$ and $g(v, u)$ denoting the sizes of subsets of $C_1(T)$ and $C_2(T)$ restricted to the connected component containing $u$ in $T - v$. Moreover, it can be done in linear time: it is sufficient to root $T$ in an arbitrary vertex, compute values of all $f(v, u)$ and $g(v, u)$ when $u$ is a child of $v$ recursively, and then by another recursion get the missing values of $f(v, parent(v))$ and $g(v, parent(v))$.
    Note that this way we can compute both size of the component (equal to $f(v, u) + g(v, u)$) as well as its imbalance (equal to $f(v, u) - g(v, u)$) on request in constant time.

    Next, let us count the total number of jumps necessary for finding central vertices over all loops in \Cref{alg:rby}. As it was stated in the proof of \Cref{lem:tree_half}, while searching for a central vertex we always jump from a vertex to its neighbor in a way that decreases the largest remaining component by one. Thus, if in the next iteration we start at exactly the neighbor of the previous central vertex, there can be only $O(n)$ such jumps in total.

    Additionally, each jump can be made in constant time, provided that we sorted all neighbors of $v$ by their subtree sizes, that is, the values of $f(v, u) + g(v, u)$ -- and that we can do also at the beginning using e.g. bucket sorting in $\deg(v)$ time per vertex (so in $O(n)$ time for all the vertices). Later, it is then sufficient to update the sizes lazily: only when we jump to certain $v$ from $u$, we fix its values $f(v, u)$ and $g(v, u)$, and move it to the proper bucket in the sorted sequence.

    Note that the rest of the $i$-th iteration of the main loop in the algorithm takes time proportional to
    \begin{align*}
        \deg(v_i) + |R_{i + 1} \cup B_{i + 1} \cup Y_{i + 1}| - |R_i \cup B_i \cup Y_i| + O(1),
    \end{align*}
    so by telescoping sum over all (at most $\lceil\log{n}\rceil$) iterations we directly obtain that the total time is also clearly bounded by $\sum_i \deg(v_i) + |R \cup B \cup Y| + O(\log{n}) \le 3 |V(T)| + O(\log{n}) = O(n)$.
\end{proof}

For completeness, we note that \Cref{lem:rby-tree} can be extended to forests:
\begin{corollary}
    \label{lem:rby-forest}
    For a forest $F = \bigcup_{i = 1}^r T_i$ on $n$ vertices consisting of $r$ trees and any $k \in [0, \sum_{i = 1}^r imb(T_i)]$ there exists a red-blue-yellow $(k, \lceil \log{n} \rceil)$-decomposition of $F$.
\end{corollary}

In order to formulate our results for the $\lambda$-backbone coloring problem, we need a simple lemma, which enables us to color a clique with a forest backbone with a given set of colors:
\begin{lemma}
    \label{lem:recolor}
    Let $F$ be a forest on $n$ vertices with a partition $V(F) = A \cup B$ into disjoint independent sets $A$ and $B$.
    Let also $[a_1, a_2]$ and $[b_1, b_2]$ be intervals such that $a_2 - a_1 \ge |A| - 1$, $b_2 - b_1 \ge |B| - 1$, $a_1 + \lambda \le b_1$ and $a_2 + \lambda \le b_2$.
    
    Then, there exists a $\lambda$-backbone coloring $c$ of $K_n$ with backbone $F$ such that $c(v) \in [a_1, a_2]$ for every $v \in A$ and $c(v) \in [b_1, b_2]$ for every $v \in B$.
    Moreover, this coloring can be found in $O(n)$ time.
\end{lemma}

\begin{proof}
    If $F$ contains only isolated vertices, then we can assign in any order colors from $[a_1, a_2]$ to $A$ and from $[b_1, b_2]$ to $B$, as it clearly would be a $\lambda$-backbone coloring with the required properties.
    
    If $F$ has a leaf $v \in B$ with a neighbor $u \in A$, then we can assign $c(v) = b_1$, $c(u) = a_1$ and invoke a subproblem for $F' = F - \{u, v\}$, $A' = A \setminus \{u\}$, $B' = B \setminus \{v\}$ with the same coloring $c$ and color intervals $[a_1 + 1, a_2]$ and $[b_1 + 1, b_2]$. The solution for $F'$ would be consistent with coloring of $u$ and $v$, as all other neighbors of $u$ in $F$ would get colors at least $b_1 + 1 \ge \lambda + a_1 + 1 > \lambda + c(u)$.
    
    Otherwise, $F$ has a leaf $v \in A$ with a neighbor $u \in B$. We can assign $c(v) = a_2$, $c(u) = b_2$ and invoke a subproblem for $F' = F - \{u, v\}$, $A' = A \setminus \{v\}$, $B' = B \setminus \{u\}$ with the same coloring $c$ and color intervals $[a_1, a_2 - 1]$ and $[b_1, b_2 - 1]$. The solution for $F'$ would be consistent with coloring of $u$ and $v$, since all other neighbors of $u$ in $F$ would get colors at most $a_2 - 1 \le b_2 - 1 - \lambda < c(u) - \lambda$.

    The linear running time follows directly from the fact that we compute $c$ only once and we can pass additionally through recursion the lists of leaves and isolated vertices in an uncolored induced subtree. The total number of updates of these lists is proportional to the total number of edges in the tree, hence the claim follows. 
\end{proof}

The lemma above can be applied to the $\lambda$-backbone coloring problem directly.
Note that here we can extend the notation $C_1$, $C_2$ used before only for trees -- however this time it is relative to a $2$-coloring $c$, which in the case of forests may not be unique:
\begin{theorem}
    \label{thm:bbc-direct}
    Let $F$ be a forest on $n$ vertices and $c$ be a $2$-coloring of $F$ such that $|C_1(F)| \ge |C_2(F)|$, where $C_i(F) = \{v \in V(F)\colon c(v) = i\}$. 
    It holds that $BBC_{\lambda}(K_n, F) \le \max\{\lambda + |C_1(F)|, n\}$. Moreover, we can find a respective coloring in $O(n)$ time.
\end{theorem}

\begin{proof}
    Let $L = \max\{\lambda, C_2(F)\}$.
    Let also $A = C_2(F)$, $B = C_1(F)$ with $[a_1, a_2] = [1, |C_2(F)|]$ and $[b_1, b_2] = [L + 1, L + |C_1(F)|$. These choices meet the conditions of \Cref{lem:recolor}, since $a_2 - a_1 = |A| - 1$, $b_2 - b_1 = |B| - 1$, $a_1 + \lambda \le a_1 + L = b_1$, and $a_2 + \lambda \le |C_2(F)| + L = |C_1(F)| + L = b_2$.
    
    Therefore, we directly infer that there exists a $\lambda$-backbone coloring $c$ using only colors no greater than $L + |C_1(F)|$.
\end{proof}

However, we can also couple \Cref{lem:recolor} with a red-blue-yellow decomposition to obtain a different bound:
\begin{theorem}
    \label{thm:bbc-decomposition}
    For a forest $F$ on $n$ vertices it holds that $BBC_{\lambda}(K_n, F) \le \max\{n, 2 \lambda\} + \Delta^2(F) \lceil\log{n}\rceil$. Moreover, we can find a respective coloring in $O(n)$ time.
\end{theorem}

\begin{proof}
    Let $L = \max\left\{\frac{n}{2}, \lambda\right\}$.
    From \Cref{lem:rby-forest} we can get a red-blue-yellow $(0, \lceil\log{n}\rceil)$-decomposition $(R, B, Y)$ of $F$.
    Note that by construction it holds that $|R| = |B| \le L$.
    
    Observe that $Y$ is not necessarily an independent set. Therefore, we need to guarantee the correctness of the coloring for the edges with both endpoints in $Y$.
    We handle that in the following way: we obtain $F'$ from $F$ by contracting all edges with both endpoints in $R \cup B$. From an arbitrary $2$-coloring $c'$ of $F'$ we construct a partition of $Y = Y_1 \cup Y_2$ such that both $Y_i = \{v \in Y\colon c'(v) = i\}$.
    Without loss of generality assume that $|Y_1| \ge |Y_2|$.
    
    Let $\N_G(U)$ be a set of neighbors of vertices from $U$ in $G$, i.e.
    \begin{align*}
        \N_G(U) = \{v \in V(G)\colon \exists_{u \in U} \{u, v\} \in E(G)\}
    \end{align*}
    Let us introduce $B_1 = B \cap \N_F(Y_1)$ and $R_1 = R \cap \N_F(B_1)$. Let also $R_2 = R \cap \N_F(Y_2)$, $B_2 = B \cap \N_F(R_2)$ and $B^* = B \setminus (B_1 \cup B_2)$, $R^* = R \setminus (R_1 \cup R_2)$.
    Note that from these definitions it follows that $B_1 \cap B_2 = R_1 \cap R_2 = \emptyset$.
    If this was not the case, then it would mean that there exists a path $(u_1, v_1, v_2, u_2)$ in $F$ such that $u_1 \in Y_1$, $u_2 \in Y_2$ and $v_1, v_2 \in R \cup B$. However, this would mean that in $F'$ both $u_1$ and $u_2$ had to be assigned the same color -- so they could not be in $Y_1$ and $Y_2$, respectively.

    To simplify the notation we define $D = \Delta(F)$.
    We will proceed with the coloring as follows:
    \begin{itemize}
        \item vertices from $Y_1$ will get colors from $1$ to $|Y_1|$,
        \item vertices from $B$ will get colors from $D |Y_1| + 1$ to $L + D |Y_1|$,
        \item vertices from $R$ will get colors from $L + D^2 |Y_1| + 1$ to $2 L + D^2 |Y_1|$,
        \item vertices from $Y_2$ will get colors from $2 L + D^2 |Y_1| + (D - 1) |Y_2| + 1$ to $2 L + D^2 |Y_1| + D |Y_2|$.
    \end{itemize}.
    We will color $F$ by assigning colors to $Y_1$, $B_1$ and $R_1$ first, and then to $Y_2$, $R_2$ and $B_2$, symmetrically.
    Finally, we prove that we are left with the part of the set $R^* \cup B^*$ that can be colored using the remaining colors (see \Cref{fig:coloring} for the assignment of colors to the sets of vertices).

    \begin{figure}[ht]
        \centering
        \includegraphics[width=0.8\textwidth]{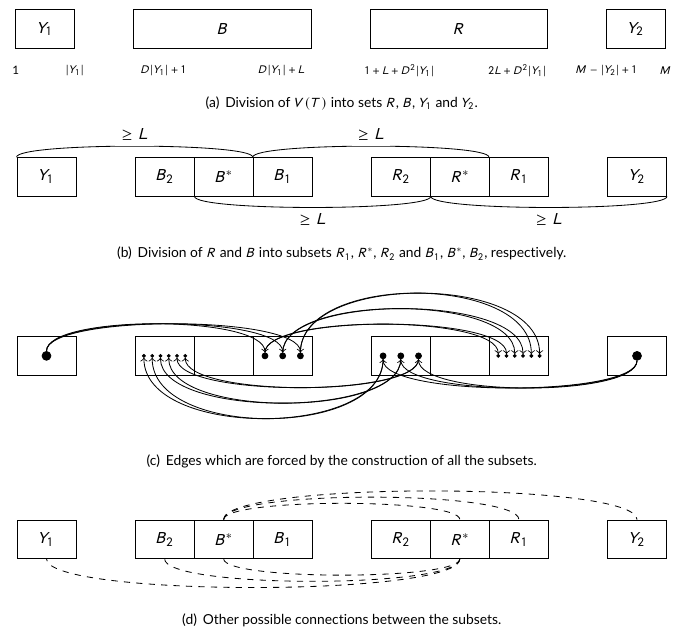}
        \caption{The main idea for the $\lambda$-backbone coloring based on a red-blue-yellow decomposition. For brevity, $M = 2 L + D^2 |Y_1| + D |Y_2|$ denotes the maximum color used.}
        \label{fig:coloring}
    \end{figure}
    
    First, let us arrange vertices from $Y_1$ in any order and assign them colors from $1$ to $|Y_1|$.
    Now, let us take vertices from $B_1$ and assign them the largest colors from the set $[D |Y_1| + 1, L + D |Y_1|]$ in the following way: the vertex with the largest colored neighbor in $Y_1$ gets $L + D |Y_1|$, second largest gets $L + D |Y_1| - 1$, etc. (ties are broken arbitrarily).
    Such a coloring ensures two things: first, all color constraints between $Y_1$ and $B_1$ are met, and second, the smallest color used for $B_1$ is at least $L + 1$.
    
    Similarly, we color $R_1$: we sort the vertices according to the largest colors of their neighbors in $B_1$ and assign the largest colors to the ones with the largest neighbor color.
    Thus we are certain that all color constraints between $B_1$ and $R_1$ are met and that the smallest color used for $R_1$ is at least $2 L + D |Y_1| + 1$, since every vertex from $B_1$ has at most $D - 1$ neighbors in $R_1$.
    
    Symmetrically, in the same fashion, we can color $Y_2$, $R_2$, and $B_2$, only this time assigning the smallest colors available for each set (see \Cref{fig:coloring}).
    
    Since there are at most $L$ vertices in both $R$ and $B$, we are sure that the sets of colors assigned to $B_1$ and $B_2$ (or $R_1$ and $R_2$, respectively) do not overlap.
    Furthermore, by construction, the only remaining vertices to be colored are exactly the ones from the set $R^* \cup B^*$.
    Note that any $v \in B^*$ can be adjacent only to some $u \in Y_2 \cup R_1$ (or to $u \in R^*$) -- but the largest color to be used for $B^*$ (certainly at most the largest color for $B$, i.e. $L + D |Y_1|$) is at least lower by $L$ from the smallest color used in $Y_2 \cup R_1$ (i.e. $2 L + D^2 |Y_1| - D (D - 1) |Y_1| = 2 L + D |Y_1|$).
    Again, the symmetrical argument goes for $R^*$ and $Y_1 \cup B_2$.

    By the argument above, the remaining colors for $B^*$ and $R^*$ form intervals $[b_1, b_2]$ and $[r_1, r_2]$ such that $b_1 + L \le r_1$, $b_2 + L \le r_2$. Thus, we are in position to color $F[R^* \cup B^*]$ with its disjoint independent sets $R^*$, $B^*$ and colors $[b_1, b_2]$, $[r_1, r_2]$ using \Cref{lem:recolor}.

    To obtain the total running time we first note that each of the initial steps -- obtaining $(R, B, Y)$ from \Cref{lem:rby-forest} (e.g. using \Cref{alg:rby}), contraction of $F$ into $F'$, and finding both $Y_1$ and $Y_2$ -- requires only linear time.
    Coloring $Y_1 \cup R_1 \cup B_1$ also requires $O(n)$ time, since we need to traverse each edge between these vertices only once to ensure the proper distances between the colors, and it is sufficient to use bucket sort to order vertices within $B_1$ and $R_1$. The same argument follows symmetrically for $Y_2 \cup R_2 \cup B_2$.
    Finally, \Cref{thm:bbc-direct} guarantees that the coloring of the remaining vertices from $R^* \cup B^*$ also can be found in $O(n)$ time.
\end{proof}

\begin{theorem}
    \label{thm:bbc-algorithm-tree}
    There exists an algorithm with running time $O(n)$ which for any forest $F$ on $n$ vertices with $\Delta(F) \ge 2$ finds a $\lambda$-backbone coloring such that $BBC_\lambda(K_n, F) \le OPT + \Delta^2(F) \lceil\log{n}\rceil$.
\end{theorem}

\begin{proof}
    First, we prove the theorem for tree backbones.
    Let us run the algorithms from \Cref{thm:bbc-direct} and \Cref{thm:bbc-decomposition} and return the better result.
    This ensures that
    \begin{align*}
        BBC_\lambda(K_n, T) & \le \min\{\max\{\lambda + |C_1(T)|, n\}, \max\{n, 2 \lambda\} + \Delta^2(T) \lceil\log{n}\rceil\} \\
        & \le \max\{\min\{\lambda + |C_1(T)|, 2 \lambda\}, n\} + \Delta^2(T) \lceil\log{n}\rceil.
    \end{align*}

    Let us also observe that if $\Delta(T) \ge 2$ (that is, if $T$ contains at least one edge) for the optimal $\lambda$-backbone coloring $c$ for a graph $K_n$ with backbone $T$ it holds that
    \begin{itemize}
        \item either $\left\lceil\frac{1}{\lambda} \max_{v \in V(T)} c(v)\right\rceil = 2$, so $\lceil\frac{c}{\lambda}\rceil$ is a $2$-coloring of $T$ -- then it is clear that $\max_{v \in V(T)} c(v) \ge \lambda + |C_1(T)|$,
        \item or $\left\lceil\frac{1}{\lambda} \max_{v \in V(T)} c(v)\right\rceil > 2$, but $\lceil\frac{c}{\lambda}\rceil$ has to be a valid coloring of $T$, but it has to use more than $2$ colors -- so $\max_{v \in V(T)} c(v) \ge 2 \lambda + 1$.
    \end{itemize}
    Therefore, it is true that $BBC_\lambda(K_n, T) \ge \min\{\lambda + |C_1(T)|, 2 \lambda + 1\}$.
    And of course we have to use a different color for each vertex, so $BBC_\lambda(K_n, T) \ge n$ -- thus $BBC_\lambda(K_n, T) \ge \max\{\min\{\lambda + |C_1(T)|, 2 \lambda + 1\}, n\}$.
    Combining all these bounds we obtain the desired result.

    To achieve the same result for forest backbones we only need to add some edges that would make the backbone connected and spanning. However, we can always make a forest connected by adding edges between some leaves and isolated vertices and we will not increase the maximum degree of the forest, as long as $\Delta(F) \ge 2$.
\end{proof}

\section{Complete graph with a tree or forest backbone: a lower bound}
\label{sec:negative}

In this section, we prove that there exists a family of trees with the maximum degree $3$ for which $BBC_\lambda(K_n, T) \ge \max\{n, 2 \lambda\} + \Theta(\log{n})$ in four steps.
First, we will define the family of trees with parameters directly related to the Fibonacci numbers. 
Next, we will show that the existence of a red-blue-yellow $(k, l)$-decomposition for such trees implies the existence of another red-blue-yellow $(k', l)$-decomposition with an additional property.
Then, we will prove that any such de\-com\-po\-si\-tion would imply further that there holds a certain decomposition of some large number (i.e. half of the large Fibonacci number) into a sum and a difference of a small number of Fibonacci numbers.
Finally, we will establish that such a decomposition cannot exist -- therefore there is no red-blue-yellow $(k, l)$-decomposition for our trees assuming that $k$ and $l$ are proportional to $\log{n}$.

We begin by connecting our red-blue-yellow $(k, l)$-decomposition to the $\lambda$-backbone coloring problem:
\begin{theorem}
    \label{thm:impossibility}
    Let $T$ be a tree on $n$ vertices. Let also $\lambda \ge 2$ and $l$ be any positive integer such that $2 \lambda + l \ge n$.
    If $BBC_\lambda(K_n, T) \le 2 \lambda + l$, them $T$ has a red-blue-yellow $(k, l)$-decomposition for some $k \in [0, \min\{\lambda - 1, 2 \lambda - n + l\}]$.
\end{theorem}

\begin{proof}
    Suppose that $T$ is a tree that does not have any red-blue-yellow decomposition for any $k \in [0, \min\{\lambda - 1, 2 \lambda - n + l\}]$, but at the same time it is true that $BBC_\lambda(K_n, T) \le 2 \lambda + l$ with some optimal $\lambda$-backbone coloring $c$.
    
    Now we can define the following sets:
    \begin{itemize}
        \item $R = \{v \in V(T)\colon 1 \le c(v) \le \lambda\}$,
        \item $B = \{v \in V(T)\colon \lambda + 1 \le c(v) \le 2 \lambda\}$,
        \item $Y = \{v \in V(T)\colon c(v) > 2 \lambda\}$.
    \end{itemize}
    Since all vertices in $c$ have different colors, it is true that $|Y| \le l$. Moreover, the optimality of $c$ implies that both $R$ and $B$ are non-empty. From the fact that $c$ is a coloring of $K_n$ it follows that both $R$ and $B$ contain at most $\lambda$ vertices since all the vertices have different colors.
    And by the definition of the backbone coloring, $R$ and $B$ have to be independent sets in $T$.
    
    Clearly, it holds that $k' = ||R| - |B|| \le \lambda - 1$.
    Moreover, since by a counting argument both $R$ and $B$ have to contain at least $n - \lambda - |Y|$ vertices (otherwise the other one would have more than $\lambda$ vertices), we know that $k' \le \lambda - (n - \lambda - |Y|) \le 2 \lambda - n + l$.
    Thus, either $(R, B, Y)$ or $(B, R, Y)$ is a red-blue-yellow $(k', l)$-decomposition of $T$ with $k' \in [0, \min\{\lambda - 1, 2 \lambda - n + l\}]$---and we obtained a contradiction.
\end{proof}

Now we define a family of \emph{Fibonacci trees}, built recursively:
\begin{definition}
    We call a rooted tree the $N$-th Fibonacci tree (denoted as $T^F_N$) if:
    \begin{itemize}
        \item $T^F_1 = T^F_2 = K_1$,
        \item $T^F_N$ for $N \ge 3$ is a tree with root $u_1$ with only child $u_2$ such that $u_2$ has two subtrees, which are, respectively, $(N - 1)$-th and $(N - 2)$-th Fibonacci trees.
    \end{itemize}
\end{definition}
We will ultimately show that this class contains infinitely many trees (starting from certain $N_0$), presented in \Cref{fig:fibonacci_trees}, cannot be colored using less than $\max\{n, 2\lambda\} + \frac{1}{48} \log_\phi{n} - 3$ colors with $\phi = \frac{1 + \sqrt{5}}{2}$.

\begin{figure}
    \centering
    \includegraphics[width=0.8\textwidth]{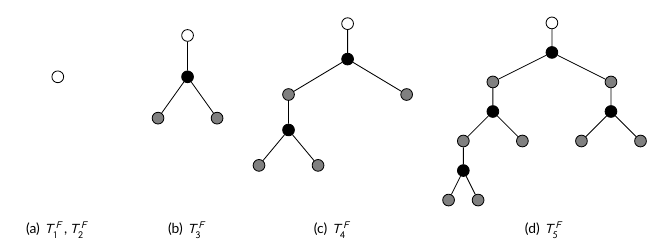}
    \caption{Fibonacci trees for $N = 1, 2, 3, 4, 5$. White vertices are the roots of Fibonacci trees and gray vertices are the roots of their Fibonacci subtrees.}
    \label{fig:fibonacci_trees}
\end{figure}


Here we will also make use of $C_1(T^F_N)$ and $C_2(T^F_N)$ sets. Without a loss of generality from now on we always define these sets with respect to a $2$-coloring $c$ of $T^F_N$ where the root of $T^F_N$ is colored with $1$.

\begin{proposition}
    \label{fact:fibonacci_subtrees}
    For every $N$ each vertex of $T^F_N$ is either the root or the parent of the root of some $T^F_i$ subtree for some $i \in [1, N]$.
\end{proposition}

\begin{theorem}
    \label{thm:fibonacci_order}
    If $T^F_N$ is the $N$-th Fibonacci tree, then $|V(T^F_N)| = 3 F_N - 2$, where $F_N$ is the $N$-th Fibonacci number.
    
    Moreover, for $n \ge 10$ it holds that $N = \lfloor\log_\phi{n}\rfloor$ (with $\phi = \frac{1 + \sqrt{5}}{2}$).
\end{theorem}

\begin{proof}
    The first part of the theorem obviously holds for $N = 1, 2$.
    If we assume that the theorem holds for all $i \in [1, N - 1]$ for some $N \ge 3$, then
    \begin{align*}
        |V(T^F_N)|
          & = |V(T^F_{N - 1})| + |V(T^F_{N - 2})| + 2 = 3 F_{N - 1} - 2 + 3 F_{N - 2} - 2 + 2 \\
          & = 3 F_N - 2.
    \end{align*}
    
    The second part of the theorem can be proved in the following way: $n = 3 F_N - 2$ implies that $\frac{\sqrt{5}}{3} (n+2) = \phi^N - (-\phi)^{-N}$.
    Therefore, we know that for $n \ge 2$ ($N \ge 3$) it holds that
    \begin{align*}
        \log_\phi\left(\phi^N - (-\phi)^{-N}\right)
              \ge N + \log_\phi\left(1 - \phi^{-2 N}\right)
              \ge N - \frac{\phi^{-2 N}}{\ln{\phi}} - \frac{\phi^{-4 N}}{\ln{\phi}} > N - \frac{1}{5},
    \end{align*}
    where we used the fact that $\ln(1 - x) \ge -x - x^2$ for $x \in (0, 0.68)$.
    
    On the other hand,
    \begin{align*}
        \log_\phi\left(\phi^N - (-\phi)^{-N}\right)
              \le N + \log_\phi\left(1 + \phi^{-2 N}\right)
              \le N + \frac{\phi^{-2 N}}{\ln{\phi}} + \frac{\phi^{-4 N}}{\ln{\phi}} < N + \frac{1}{5},
    \end{align*}
    where we used the fact that $\ln(1 + x) \le x + x^2$ for all $x > 0$.
    
    Finally, for $n \ge 10$ ($N \ge 5$) it holds that
    \begin{align*}
        \log_\phi\left(\frac{\sqrt{5}}{3} (n+2)\right)
              & = \log_\phi{n} + \log_\phi\left(\frac{\sqrt{5}}{3}\right) + \log_\phi\left(1 + \frac{2}{n}\right) \\
              & \le \log_\phi{n} - \frac{3}{5} + \frac{2}{n \log{\phi}} + \frac{4}{n^2 \log{\phi}} < \log_\phi{n} - \frac{1}{5},
    \end{align*}
    and $\log_\phi\left(\frac{\sqrt{5}}{3} (n+2)\right) > \log_\phi{n} + \log_\phi\left(\frac{\sqrt{5}}{3}\right) > \log_\phi{n} - \frac{4}{5}$.
    
    By putting all these bounds together we obtain that $N < \log_\phi{n} < N + 1$, which completes the proof.
\end{proof}

\begin{theorem}
    \label{thm:fibonacci_imb}
    If $T^F_N$ is the $N$-th Fibonacci tree, then $imb(T^F_N) = F_N$, where $F_N$ is $N$-th Fibonacci number.
\end{theorem}

\begin{proof}
    Clearly $imb(T^F_N) = F_N$ for $N = 1, 2$.
    
    Let $N \ge 3$ and proceed by induction. By construction, we know that the grandchildren of the root of $T^F_N$ are exactly the roots of $T^F_{N - 1}$ and $T^F_{N - 2}$ and they all have the same color in any $2$-coloring of $T^F_N$.
    Therefore,
    \begin{align*}
        imb(T^F_N)
          & = imb(T^F_{N - 1}) + imb(T^F_{N - 2}) = F_{N - 1} + F_{N - 2} = F_N.
    \end{align*}
\end{proof}

Throughout the remainder of the paper we denote by $parent(v)$ a vertex directly above $v$ in $T^F_N$. For consistency, we assume that if $r$ is a root of $T^F_N$, then $parent(r)$ is properly defined so that the predicate $parent(r) \in A$ is false for any set $A$.

Observe that the set of vertices which are the roots of all Fibonacci subtrees is defined exactly by $C_1(T^F_N)$, since the root of $T^F_N$ is in $C_1(T^F_N)$.
Therefore, using \Cref{fact:fibonacci_subtrees} we can modify any red-blue-yellow $(k, l)$-decomposition in such a way that $Y \cap C_1(T^F_N) = \emptyset$:
\begin{lemma}
    \label{lem:pushdown}
    Let $(R, B, Y)$ be a red-blue-yellow $(k, l)$-decomposition of a Fibonacci tree $T^F_N$.
    There exists $0 \le k' \le k + 2 l$ such that we can construct a red-blue-yellow $(k', l)$-decomposition $(R', B', Y')$ of $T^F_N$ with $Y' \cap C_1(T^F_N) = \emptyset$.
\end{lemma}

\begin{proof}
    First, let us define $Y' = \{v \in C_2(T^F_N)\colon v \in Y \text{ or } parent(v) \in Y\}$, i.e. let us ``push down'' the yellow color from roots of Fibonacci subtrees in $T^F_N$ to their children. Clearly, $Y' \cap C_1(T^F_N) = \emptyset$ and moreover $|Y'| \le |Y|$ since the construction implicitly defines a function from $Y$ into $Y'$. 
    
    Next, we define $R'$ and $B'$ to be as following:
    \begin{itemize}
        \item if $v \in C_2(T^F_N)$, $v \notin Y$ and $parent(v) \notin Y$, then we preserve the color: if $v \in R$, then we put $v$ in $R'$, otherwise we put it in $B'$,
        \item if $v \in C_1(T^F_N)$ and $v \notin Y$, then we also preserve the color of $v$,
        \item if $v \in C_1(T^F_N)$, $v \in Y$ and $parent(v) \notin Y$, then we use the color different than in the parent: if $parent(v) \in R$, then we add $v$ to $B'$, otherwise we add $v$ to $R'$,
        \item finally, if $v \in C_1(T^F_N)$, $v \in Y$ (so by definition the only child of $v$ is in $Y'$) and $parent(v) \in Y$, we put $v$ in $R'$.
    \end{itemize}
    Intuitively, this whole procedure ``extends'' sets of blue and red vertices to the previously yellow vertices $v \in Y \cap C_1(T^F_N)$: if $parent(v)$ was blue (respectively, red), we use red (respectively, red) to color $v$. 
    And if we are certain that both a parent and an only child of a vertex in $C_1(T^F_N)$ will be in $Y'$ (the last case above), then we just assign it to $R'$, as it will be surrounded only by the vertices from $Y'$ (in fact, it could be also added to $B'$ instead).
    
    Note that $|R'| \le |R| + |Y|$ as only in the last two cases we can add some vertices to $R'$ which were not in $R$ -- but this means that these vertices were in $Y$.
    Additionally, $|B'| \ge |B| - |Y'| \ge |B| - |Y|$, as the only vertices which are in $B \setminus B'$ are exactly the ones which are added to $Y'$.
    
    Finally, we verify that indeed $(R', B', Y')$ or $(B', R', Y')$ is a red-blue-yellow $(k', l)$-decomposition of $T^F_N$ for some
    \begin{align*}
        k' = ||R'| - |B'|| \le ||R| - |B| + 2 |Y|| \le k + 2 l,
    \end{align*}
    because $R'$ and $B'$ are independent sets and $|Y'| \le |Y| \le l$.

    To see that $R'$ is an independent set, let us note that if there were some $v, parent(v) \in R'$, then $v, parent(v) \notin B$, since we never swap between red and blue colors in the procedure above.
    Thus we could have only five cases:
    \begin{itemize}
        \item $v, parent(v) \in Y$ -- but since one of them is in $C_2(T_N^F)$, it has to be also in $Y'$, not in $R'$,
        \item $v \in Y \cap C_1(T_N^F)$, $parent(v) \in R$ -- but then $v$ would be rather in $B'$, not in $R'$.
        \item $v \in Y \cap C_2(T_N^F)$, $parent(v) \in R$ -- thus $v$ would be in $Y'$, not in $R'$,
        \item $v \in R \cap C_1(T_N^F)$, $parent(v) \in Y$ -- then by construction $parent(v)$ would be in $Y'$, not in $R'$,
        \item $v \in R \cap C_2(T_N^F)$, $parent(v) \in Y$ -- and $v$ would be in $Y'$, not in $R'$.
    \end{itemize}
    Either way, we would obtain a contradiction.
    Note that the same argument can be also used to show that $B'$ is an independent set, which concludes the proof.
\end{proof}

Now, we turn to the computation of the imbalances of red-blue parts of Fibonacci trees.
In order to do that, we have to recall some theory of Fibonacci numbers.
In particular, by the Zeckendorf theorem \cite{knuth1988fibonacci} we know that any number $n \in \mathbb{N}$ can be written as $n = \sum_j z_j F_j$ for $z_j \in \{0, 1\}$ such that
\begin{enumerate}
    \item[(A)] no two consecutive $z_j$ are equal to $1$, i.e. $z_j z_{j + 1} = 0$ for all $j$,
    \item[(B)] this representation is unique and it is the decomposition with the smallest value of $\sum_j |z_j|$,
    \item[(C)] it can be constructed by greedy subtraction of the largest possible Fibonacci numbers from $n$ until we reach $0$.
\end{enumerate}
We call a vector $(z_j)_{j = 1}^\infty$ such that $(A)$--$(C)$ holds the \emph{Zeckendorf representation} of $n$ and denote it by $Z(n)$.

In fact, we can generalize the property $(B)$ to $(B')$ and show that for any $n = \sum_j z'_j F_j$ with $z'_j \ge 0$ it has to be the case that $\sum_j |z'_j| \ge \sum_j |z_j|$ for $Z(n) = (z_j)_{j = 1}^\infty$.

Without loss of generality we trim $Z(n)$ to a finite number of positions with a leading one since $F_i > n$ implies trivially $z_i = 0$.
For example, $Z(F_N) = (0, \ldots, 0, 1)$, i.e. it has a single one on $N$-th position and $Z\left(\frac{F_N}{2}\right) = (\ldots, 1, 0, 0, 1)$, i.e. it has ones only on every third position ending at $N - 2$ as it is true that $\frac{F_N}{2} = F_{N - 2} + \frac{F_{N - 3}}{2}$.

Now we may proceed to the lemma which ties tree decompositions and Fibonacci numbers:
\begin{lemma}
    \label{lem:count}
    For any red-blue-yellow $(k, l)$-decomposition $(R, B, Y)$ of $T^F_N$ such that $Y \cap C_1(T^F_N) = \emptyset$ and $k \ge 0$ it holds that every connected component $T_i$ of $T^F_N[R \cup B]$ has
    \begin{align*}
        imb(T_i) = |Y_i| + \sum_{j = 1}^r z_{i, j} F_j
    \end{align*}
    where $r \in [1, N]$, $Y_i = \{v \in Y\colon parent(v) \in V(T_i)\}$, $z_{i, j} \in \{-1, 0, 1\}$ for all $i$ and $j = 1, 2, \ldots, N$ and $\sum_{j = 1}^N |z_{i, j}| \le |Y_i|$.
\end{lemma}

\begin{proof}
    The intuition behind the theorem is straightforward:
    since $Y \cap C_1(T^F_N) = \emptyset$, it follows that each component $T_i$ of $T^F_N[R \cup B]$ is a tree rooted in a vertex $u_i$ that is also a root of some Fibonacci tree. It may be a Fibonacci tree itself -- but even if it is not, then it means it has some Fibonacci subtrees removed, which are cut off exactly in vertices from $Y_i$.
    But then every $v \in Y_i$ has children that are also roots of some Fibonacci subtrees.

    Now, let us define $c$ as a $2$-coloring of $T^F_N$ and $c_i$ as a $2$-coloring of $T_i$, such that all their respective roots have color $1$. From now on, $C_1(T^F_N)$ and $C_2(T^F_N)$ are defined with respect to $c$.
    
    Observe that it has to be the case that $c(u_i) = 1$, since every root of a Fibonacci subtree is at the even distance from the root of the Fibonacci tree $T^F_N$. Thus, for all $v \in V(T_i)$ it holds that $c_i(v) = c(v)$.
    
    Moreover, every leaf $v$ of $T_i$ also has $c_i(v) = 1$, as either it is a leaf of $T^F_N$ (thus $c(v) = 1$), or it is a  parent of some vertex from $Y_i$. However, in the latter case $Y \cap C_1(T^F_N) = \emptyset$ implies also that $c(v) = 1$, as the vertices from $Y_i$ always get color $2$.
    
    Finally, we note that to get $imb(T_i)$ we have to:
    \begin{itemize}
        \item add the imbalance of a Fibonacci subtree rooted in $u_i$, i.e. $T^F_r$ for some $r \in [1, N]$,
        \item remove the imbalance of some Fibonacci subtrees of $T^F_r$ (i.e. the ones rooted in children of all $v \in Y_i$ below $u_i$) -- note that for each vertex from $Y_i$ we remove its two subtrees $T_{j - 1}^F$ and $T_{j - 2}^F$ with total imbalance equal to $imb(T_{j - 1}^F) + imb(T_{j - 2}^F) = F_j$,
        \item add the size of $|Y_i|$ -- vertices from $Y_i$ were taken into account when adding the imbalance of the whole tree in the first case above, however, they are not in $T_i$. Note that their removal will increase the imbalance since they are always in $C_2(T_N^F)$.
    \end{itemize}
    See also an example in \Cref{fig:fibonacci}.

    \begin{figure}
        \centering
    \includegraphics[width=0.8\textwidth]{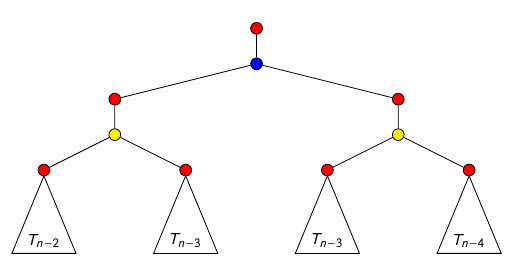}
        \caption{An example Fibonacci tree $T^F_n$ with $|Y_1| = 2$ and the imbalance of the top tree $imb(T_1) = |Y_1| + F_n - F_{n - 1} - F_{n - 2} = 2$.}
        \label{fig:fibonacci}
    \end{figure}
    
    Overall, if the imbalance of a Fibonacci subtree of $T^F_N$ rooted in $u_i$ is equal to $F_r$ for some $r$, then we get that $imb(T_i) = F_r - \sum_{j = 1}^{r - 1} z'_{i, j} F_j + |Y_i|$ for some $z'_{i, j} \in \mathbb{N}$ with $\sum_{j = 1}^{r - 1} |z'_{i, j}| \le |Y_i|$. Here each $z'_{i, j}$ denotes the number of Fibonacci subtrees $T_j^F$ which had to be cut out from $T_r^F$ to obtain exactly $T_i$.

    Finally, we can replace the last element using the fact that the definition of the Zeckendorf representation $Z(\sum_{j = 1}^{r - 1} z'_{i, j} F_j) = (-z_{i, j})_{j = 1}^{r - 1}$ implies $\sum_{j = 1}^{r - 1} z'_{i, j} F_j = \sum_{j = 1}^{r - 1} -z_{i, j} F_j$ for some $z_{i, j} \in \{-1, 0\}$. And from the property $(B')$ it follows that $\sum_{j = 1}^{r - 1} |z_{i, j}| \le \sum_{j = 1}^{r - 1} |z'_{i, j}| \le |Y_i|$.
    Therefore, it is sufficient to set $z_{i, r} = 1$ to conclude the proof. 
\end{proof}

\begin{lemma}
    \label{lem:split}
    If there exists a red-blue-yellow $(k, l)$-decomposition $(R, B, Y)$ of $T^F_N$ such that $|Y| = l$, $Y \cap C_1(T^F_N) = \emptyset$, then there exists a value $y \in [0, l]$ and a vector $z = (z_j)_{j = 1}^N$ with $z_j \in \{0, 1\}$, $\sum_{j = 1}^N |z_j| \le 2 l + 1$, and
    \begin{align*}
        \frac{F_N - k + l}{2} = y + \sum_{j = 1}^N z_j F_j.
    \end{align*}
\end{lemma}

\begin{proof}
    The required decomposition exists if and only if there exists a set $Y$ with $l$ vertices such that we can find for connected components $T_i$ of $T_N^F[R \cup B]$ certain $2$-coloring $c$ which induces a partition into $R$ and $B$ such that $|R| - |B| = k$. 
    
    From \Cref{lem:count} we know that for every $T_i$ we can write $imb(T_i) = |Y_i| + \sum_{j = 1}^N z_{i,j} F_j$ such that $z_{i,j} \in \{-1, 0, 1\}$ and $\sum_{j = 1}^N |z_{i,j}| \le |Y_i| + 1$. Thus,
    \begin{align*}
        k = |R| - |B|
            & = \sum_{i = 1}^{l + 1} a_i\,imb(T_i) = y + \sum_{i = 1}^{l + 1} \sum_{j = 1}^N a_i z_{i,j} F_j.
    \end{align*}
    for some $y = \sum_{i = 1}^{l + 1} a_i |Y_i| \in \left[-l, l\right]$ and some $a_i \in \{-1, 1\}$.
    
    Intuitively, we add or subtract imbalances of $T_i$ depending on the fact whether the root of $T_i$ has color $1$ or $2$ in $c$.
    In particular, $a_i = -1$ if and only if the root of $T_i$ has color $2$ in $c$.
    Note that from $|Y| = l$ we know that $T_N^F[R \cup B] = T_N^F \setminus Y$ has exactly $l + 1$ connected components (some possibly empty).

    Now let us consider a coloring $c'$ of $T_N^F$ where the root has color $1$.
    Note that all subtrees $T_i$ are rooted in vertices with color $1$ in $c'$, and all vertices from $Y$ have color $2$ in $c'$.
    Therefore,
    \begin{align*}
        F_N + l
            & = imb(T_N^F) + |Y| = \sum_{i = 1}^{l + 1} imb(T_i) = |Y| + \sum_{i = 1}^{l + 1} \sum_{j = 1}^N z_{i,j} F_j,
    \end{align*}
    where the last equality follows from \Cref{lem:count}.
    
    Now, by subtracting the last two equations we obtain
    \begin{align*}
        \frac{F_N - k + l}{2} = \frac{|Y| - y}{2} + \sum_{i = 1}^{l + 1} \sum_{j = 1}^N a'_i z_{i,j} F_j,
    \end{align*}
    where we used the fact that $a'_i = \frac{1 - a_i}{2} \in \{0, 1\}$ for all $i = 1, 2, \ldots, l + 1$.
    
    Now observe that $\sum_{j = 1}^N F_j \sum_{i = 1}^{l + 1} a'_i z_{i,j}$ is a sum of Fibonacci numbers, i.e. each $F_j$ appears exactly $\sum_{i = 1}^{l + 1} a'_i z_{i,j}$ times. The total count of Fibonacci numbers in this sum does not exceed
    \begin{align*}
        \sum_{i = 1}^{l + 1} \sum_{j = 1}^N |a'_i z_{i,j}| \le \sum_{i = 1}^{l + 1} \sum_{j = 1}^N |z_{i,j}| \le \sum_{i = 1}^{l + 1} (|Y_i| + 1) \le |Y| + l + 1 = 2 l + 1,
    \end{align*}
    since we know from \Cref{lem:count} that $\sum_{j = 1}^N |z_{i,j}| \le |Y_i| + 1$ for every $i$.
    
    Therefore, by Zeckendorf theorem there exists a vector $z' = (z'_j)_{j = 1}^N$ with $z'_j \in \{0, 1\}$ so that $\sum_{j = 1}^N |z'_j| \le 2 l + 1$ and
    \begin{align*}
        Z\left(\sum_{i = 1}^{l + 1} \sum_{j = 1}^N a'_i z_{i,j} F_j\right) = z'.
    \end{align*}
    This completes the proof.
\end{proof}

\begin{lemma}
    \label{lem:fibonacci}
    For $N \ge 96$, any $k \in [0, F_{N / 2 - 1}]$, $l \in \left[1, \frac{N}{48} - 1\right]$ there does not exist a value $y \in [0, l]$ and a vector $z = (z_j)_{j = 1}^N$ with $z_j \in \{-1, 0, 1\}$ and $\sum_{j = 1}^N |z_j| \le 2 l + 1$ such that
    \begin{align*}
        \frac{F_N - k + l}{2} = y + \sum_{j = 1}^N z_j F_j.
    \end{align*}
\end{lemma}

\begin{proof}
    It is well known that $F_{N - 2} < \frac{F_N}{2} < F_{N - 1}$ for any $N \ge 4$.
    Moreover, it follows from induction that
    \begin{align*}
        \frac{F_N}{2} = \sum_{i = 0}^{\lfloor\frac{N - 2}{3}\rfloor} F_{N - 2 - 3 i}.
    \end{align*}
    
    From now on let us denote $K = F_{N / 2 - 1}$ an $L = \frac{N}{48} - 1$.
    In fact, we would like to prove a slightly stronger claim: that any number $S \in \left[\frac{F_N - K - L}{2}, \frac{F_N + L}{2}\right]$ (a range encompassing all possible values of $\frac{F_N - k + l}{2} - y$ since $y \le l$) cannot be represented using a sum or difference of at most $2 l + 1$ Fibonacci numbers.

    Let $S' = S - \frac{F_N}{2}$.
    By assumption it holds that $L < K = F_{N / 2 - 1}$, so $|S'| \le \frac{K + L}{2} < F_{N / 2 - 1}$ and therefore $Z(|S'|)$ cannot have ones on positions $\left[\frac{N}{2} - 1, N\right]$.
    Let us take positions $\left[1, \frac{N}{2} - 1\right]$ of $Z\left(\frac{F_N}{2}\right)$, equivalent to some number $S''$ and observe that $S'' + S' \ge 0$ -- and, additionally, it also cannot have ones on positions $\left[\frac{N}{2} - 1, N\right]$.
    Therefore, we are certain that $Z(S)$ has exactly the same values as $Z\left(\frac{F_N}{2}\right)$ at least on positions $\frac{N}{2} + 1, \ldots, N$.

    Note that we can split $Z\left(\frac{F_N}{2}\right)$ into a sequence of $\lfloor\frac{N - 1}{3}\rfloor$ blocks of length $3$ equal to $(0, 1, 0)$. For convenience, but with a slight abuse of notation from now on we allow for a constant number of leading zeros in the Zeckendorf representations.
    Combining the last two facts we get that $Z(S) = Z\left(\frac{F_N}{2} + S'\right)$ contains a sequence of $\frac{N}{6}$ consecutive blocks $(0, 1, 0)$ between positions $\frac{N}{2} + 1$ and $N$.

    Now let us define $S_1 = \sum_{j = 1}^N 1_{z_j > 0} \cdot F_j$ and $S_2 = \sum_{j = 1}^N 1_{z_j < 0} \cdot F_j$. Clearly, $S = S_1 - S_2$ and $Z(S_1)$ and $Z(S_2)$ have both in total at most $2 l + 1$ ones.
    Therefore, to prove this lemma it is sufficient to show that if $Z(S)$ consists of $\frac{N}{6}$ consecutive blocks $(0, 1, 0)$ and $Z(S_2)$ has only $2 l + 1$ ones, then $Z(S_1)$ cannot have at most $2 l + 1$ ones.
    From now on we will concentrate only on these sections of the vectors $Z(S)$ and $Z(S_2)$, and the possible carry operations in the addition $Z(S) + Z(S_2) = Z(S_1)$.
    
    Let us write both $Z(S)$ and $Z(S_2)$ as blocks of length $3$ and denote $i$-th block of $Z(S_2)$ (and also, by association, the respective block of $Z(S)$) as type A if it is non-zero, and as type B otherwise.
    Note that there can be at most $2 l + 1$ blocks of type A since $Z(S_2)$ has at most $2 l + 1$ ones.
    
    Now consider how many blocks of type B can be influenced by a carry from blocks of type A. We have to analyze not only the forward carry but also the backward one, as for example, in Zeckendorf representation we have $Z(3) = (0, 0, 0; 1, 0, 0)$ but $Z(3 + 3) = (0, 1, 0; 0, 1, 0)$.
    
    The proof for the forward carry is simpler: suppose we consider the $i$-th block of type A, which starts with a position $j$, and that the $(i + 1)$-th block is of type B.
    Then we know that the number $S^*_2$ indicated only by the blocks up to $i - 1$ cannot exceed $F_j - 1$ by the definition of Zeckendorf representation. Similarly, for $Z(S)$ the respective number $S^*$ cannot exceed $F_{j - 1} - 1$ since its $(i - 1)$-th block has to be equal to $(0, 1, 0)$.
    Now, the forward carry has three components:
    \begin{itemize}
        \item the forward carry from the sum of $S^*$ and $S_2^*$ to the $i$-th block -- which does not exceed $(F_j - 1) + (F_{j - 1} - 1) < F_{j + 1}$, so it can be only in the form $(1, 0, 0)$,
        \item the value of the $i$-th block of $Z(S)$ -- equal to $(0, 1, 0)$,
        \item the value of the $i$-th block of $Z(S_2)$ -- at most equal to $(1, 0, 1)$ since it does not have two consecutive ones.
    \end{itemize}
    In any way, the forward carry to the $(i + 1)$-th block cannot exceed $(1, 1, 0)$. However, since the $(i + 1)$-th blocks of $Z(S)$ and $Z(S_2)$ are $(0, 1, 0)$ and $(0, 0, 0)$, respectively (since it is of type B), there is no forward carry to the $(i + 2)$-th block and beyond.
    
    That leaves us with proving that backward carry from a block of type A cannot influence too many blocks of type B.
    First, we note that $Z(S_2)$ by the property $(A)$ of the Zeckendorf representation does not have two consecutive ones. Thus, the only combinations available when we sum the rightmost blocks of type A (i.e. the ones which do not have blocks of type A to the right) are:
    \begin{enumerate}
        \item $(0,1,0) + (1,0,0) = (0,0,1)$,
        \item $(0,1,0) + (0,1,0) = (0,0,1)$ with a backward carry $(0, 0, 1)$,
        \item $(0,1,0) + (0,0,1) = (0,0,0)$ with a forward carry $(1, 0, 0)$,
        \item $(0,1,0) + (1,0,1) = (1,0,0)$ with a forward carry $(1, 0, 0)$.
    \end{enumerate}
    The only possible backward carry happens in the second case and it is equal to a single one on the rightmost position of the carry.
    
    Now, observe that if the block to the left is also of type A, then a respective block from $Z(S)$ is $(0, 1, 0)$ -- and when we add the backward carry $(0, 0, 1)$ to it, we obtain the forward carry to the rightmost block. And regardless of the value of the appropriate block of $Z(S_2)$, the total sum of the blocks and the backward carry cannot generate any further backward carry.
    Finally, note that the aforementioned forward carry resulting from backward carry appears in the block which has to be equal to $(0, 0, 1)$ (as it has to be the second case above), so it turns it into $(1, 0, 1)$ and it does not generate any future carries.
    
    Therefore, the only possible backward carry from the block of type A to the block of type B has to be in the form $(0, 0, 1)$. However, this will be combined with a block $(0, 1, 0)$ from $Z(S)$ -- thus, the sum of the blocks from $Z(S)$, $Z(S_2)$ and the backward carry again cannot result in any further backward carry.

    Thus, we proved that any block of type A can affect up to one block of type B to the left and one block of type B to the right.
    If there are only $2 l + 1$ blocks of type A, then at most $3 (2 l + 1)$ blocks of type A and B can be modified -- and therefore at least $\frac{N}{6} - 3 (2 l + 1)$ blocks of the form $(0, 1, 0)$ of $Z(S - S_2)$ are identical to the respective blocks of $Z(S)$ despite the addition of $Z(S_2)$.
    This means that $Z(S - S_2)$ has to have at least $\frac{N}{6} - 3 (2 L + 1)$ ones -- which contradicts the assumption that $Z(S - S_2) = Z(S_1)$ has at most $2 L + 1$ ones for $L = \frac{N}{48} - 1$.
\end{proof}

\begin{theorem}
    \label{thm:nonexistence}
    For any $N \ge 96$, any $k \in \left[0, F_{N / 2 - 1} - \frac{N}{24} + 2\right]$ and any $l \in \left[1, \frac{N}{48} - 1\right]$ there does not exist a red-blue-yellow $(k, l)$-decomposition of $T_N^F$.
\end{theorem}

\begin{proof}
    A direct implication of \Cref{lem:split} and \Cref{lem:fibonacci} is that there does not exist a red-blue-yellow $(k', l)$-de\-com\-po\-si\-tion of $T_N^F$ such that $Y \cap C_1(T^F_N) = \emptyset$ for any $k' \in [0, F_{N / 2 - 1}]$ and $l = \frac{N}{48} - 1$.
    
    Now, it follows from \Cref{lem:pushdown} that if there was a red-blue-yellow $(k, l)$-decomposition of $T_N^F$ for $k$, then it would exist a red-blue-yellow $(k', l)$-de\-com\-po\-si\-tion of $T_N^F$ with $Y \cap C_1(T^F_N) = \emptyset$ and $k' \in [0, k + 2 l] \subseteq [0, F_{N / 2 - 1}]$ -- thus, a contradiction.
\end{proof}

Finally, we can proceed to the main theorem of this section, which shows the existence of the trees with constant degrees such that $BBC_\lambda(K_n, T) = \max\{n, 2\lambda\} + \Omega(\log{n})$:
\begin{theorem}
    There exists an infinite family of trees $T_N^F$ on $n = 3 F_N - 2$ vertices with $\Delta(T_N^F) = 3$ for $N \ge 96$ such that $BBC_\lambda(K_n, T_N^F) \ge \max\{n, 2\lambda\} + \frac{1}{48} \log_\phi{n} - 3$ for $\lambda = \left\lfloor\frac{n}{2}\right\rfloor$.
\end{theorem}

\begin{proof}
    Assume that $l = \frac{N}{48} - 1$.
    Then it always holds that $2 \lambda - n + l \le \lambda - 1$ -- so in \Cref{thm:impossibility} we need that there cannot be any red-blue-yellow $(k, l)$-decomposition for $k \in [0, 2 \lambda - n + l]$.
    If $N \ge 96$, then it is easy to check that $F_{N / 2 - 1} - \frac{N}{24} + 2 > 2 \lambda - n + l$.
    Therefore, by combining \Cref{thm:nonexistence,thm:impossibility} we get that $BBC_\lambda(K_n, T_N^F) > 2 \lambda + \frac{N}{48} - 1$.
    To finish the proof it is sufficient to combine this result with two simple observations:
    \begin{itemize}
        \item by assumption it holds that $2\lambda \ge \max\{n, 2\lambda\} - 1$,
        \item by \Cref{thm:fibonacci_order} it is true that $N = \lfloor\log_\phi{n}\rfloor > \log_\phi{n} - 1$.
    \end{itemize}
\end{proof}
In fact, the theorem above can be proven in a more complicated, but also more general form, i.e. that $BBC_\lambda(K_n, T_N^F) > 2 \lambda + \frac{1}{48} \log_\phi{n} - 2$ for all $\lambda \in \left[\frac{n}{2} - \frac{1}{96} \log_\phi{n}, \frac{n}{2}\right]$ -- which gives us also an additive factor over $\max\{n, 2\lambda\}$.

\section{Conclusion}
\label{sec:open}

In this paper, we presented a new bound on $BBC_\lambda(K_n, F) = \max\{n, 2 \lambda\} + \Delta^2(F)\lceil\log{n}\rceil$ for forests $F$ on $n$ vertices. This is the first known upper bound depending on the maximum degree of $F$, better than the existing ones when this parameter is small, e.g. constant.
We also showed that the same reasoning allows us to construct a linear algorithm for finding an $\lambda$-backbone coloring of $K_n$ with backbone $F$ within an additive error of $\Delta^2(F)\lceil\log{n}\rceil$ from the optimum.
Finally, we proved that the bound is asymptotically tight, i.e. there exists a family of trees with $\Delta(T) = 3$ such that $BBC_\lambda(K_n, F) = \max\{n, 2 \lambda\} + \Theta(\log{n})$.

However, there remain a couple of open problems that stem from our research.
First, we can ask if the algorithm from \Cref{sec:positive} can be improved, i.e. does there exist an algorithm running in polynomial time which gives us a coloring using $BBC_\lambda(K_n, F) + o(\log{n})$ colors.
Or maybe it is the case that finding a $\lambda$-backbone coloring such that $BBC_\lambda(K_n, F) + f(n, \Delta(F))$ for some function $f$ is $\NP$-hard -- either for $\Delta(F)$ constant or provided as a part of input.

Second, there arises a natural question of how to extend our results to other classes of graphs.
An obvious extension would be an analysis for a class of split graphs, i.e. graphs whose vertices can be partitioned into a maximum clique $C$ (of size $\omega(G) = \chi(G)$) and an independent set $I$.
A simple application of \Cref{thm:bbc-algorithm-tree} gives us that
\begin{align*}
    BBC_\lambda(G, F) & \le \lambda + BBC_\lambda(G[C], F[C]) 
         \le \lambda + \max\{\chi(G), 2 \lambda\} + \Delta^2(F) \lceil\log{\chi(G)}\rceil,
\end{align*}
as we can first solve the problem restricted to $C$, find a $\lambda$-backbone coloring using colors from the set $[1, \max\{\chi(G), \allowbreak 2 \lambda\} + \Delta^2(F) \lceil\log{\chi(G)}\rceil]$, and then color all vertices in $I$ with a color $\lambda + \max\{\chi(G), 2 \lambda\} + \Delta^2(F) \lceil\log{\chi(G)} \rceil$.

However, it was proved before that
\begin{theorem}[Salman, \cite{salman2006lambda}]
If $G$ is a split graph and $T$ is its spanning tree, then
\begin{align*}
    BBC_{\lambda}(G, T) \le
    \begin{cases}
        1 & \text{for $\chi(G) = 1$,} \\
        \lambda + 1 & \text{for $\chi(G) = 2$,} \\
        \lambda + \chi(G) & \text{for $\chi(G) \ge 3$.}
    \end{cases}
\end{align*}
This bound is tight.
\end{theorem}
Although our bound is more general, as it applies to all backbones $H$ as long as $H[C]$ is a forest, a simple comparison of both bounds shows that for trees (and forests) our bound is worse in all cases, even for small $\lambda$ and large $\chi(G)$. This is mainly because we assign a single color to $I$ naively instead of trying to optimize it.
Still, we can pose a question: is there a way to amend our approach to make it suitable also for split graphs and improve the existing bound at least for some subset of graphs?

\bibliography{backbone.bib}

\end{document}